\documentclass[a4paper,11pt,reqno]{article}

\usepackage{graphicx,amssymb}
\usepackage{amssymb,amsmath}
\usepackage[colorlinks]{hyperref}
\usepackage{amsfonts}
\usepackage{amscd}
\usepackage{amsthm}
\usepackage{caption}
\usepackage{cite}

\numberwithin{equation}{section}
\newtheorem{thm}{Theorem}[section]
\newtheorem{Lemma}{Lemma}[section]

\newtheorem{Def}{Definition}[section]
\newtheorem{rem}{Remark}[section]

\usepackage[font=small,skip=0pt]{caption}

\newcommand\blfootnote[1]{%
	\begingroup
	\renewcommand\thefootnote{}\footnote{#1}%
	\addtocounter{footnote}{-1}%
	\endgroup
}

\begin{document}
	
\title{Modeling the dynamics of the Hepatitis B virus via a variable-order discrete system\blfootnote{This
is a preprint whose final form is published in \emph{Chaos, Solitons and Fractals} 
\url{https://doi.org/10.1016/j.chaos.2024.114987}.}}

\baselineskip 16pt

\author{Meriem Boukhobza$^{a,}$\footnote{marboukhobza@gmail.com}\ ,  
Amar Debbouche$^{b,d,}$\footnote{amar\_debbouche@yahoo.fr}\ ,\\ 
Lingeshwaran Shangerganesh$^{c,}$\footnote{shangerganesh@nitgoa.ac.in}
\ and\  
Delfim F. M. Torres$^{d,}$\footnote{delfim@ua.pt}
\medskip \\
{\small $^a$Department of Mathematics and Informatics, University of Mostaganem, Mostaganem 27000, Algeria}\\
{\small $^b$Department of Mathematics, Guelma University,  Guelma 24000, Algeria}\\
{\small $^c$Department of Applied Sciences, National Institute of Technology Goa, Goa 403703, India}\\
{\small $^d$R\&D Unit CIDMA, Department of Mathematics, University of Aveiro, Aveiro 3810-193, Portugal}}

\date{\empty}

\maketitle	

% -----------------------------------------------------

\begin{abstract}
We investigate the dynamics of the hepatitis B virus by integrating 
variable-order calculus and discrete analysis. 
Specifically, we utilize the Caputo variable-order difference operator 
in this study. To establish the existence and uniqueness results of the model, 
we employ a fixed-point technique. Furthermore, we prove that the model exhibits 
bounded and positive solutions. Additionally, we explore the local stability 
of the proposed model by determining the basic reproduction number. Finally, 
we present several numerical simulations to illustrate the richness of our results.
\end{abstract}

% -----------------------------------------------------

\section{Introduction}

The body experiences a diminished immune response upon exposure to viruses, 
particularly viral hepatitis, as a defense mechanism to protect the liver. 
There are various types of hepatitis, including hepatitis A, hepatitis B, 
hepatitis C, hepatitis D, and hepatitis E. Our paper specifically 
concentrates on hepatitis B. The hepatitis B virus is recognized as one 
of the most harmful viruses, significantly contributing to the development 
of liver cancer \cite{csener2018investigation}. Its rapid transmission 
is evident, with approximately $80\%$ of hepatitis B cases attributed to 
infections transmitted from person to person through blood and sexual contact
\cite{pol2005epidemiologie,simelane2021fractional}. Horizontal transmission 
from mother to fetus is also documented 
\cite{khan2021transmission,thornley2008hepatitis}.

Hepatitis B manifests in two forms: acute and chronic. Many carriers of the virus, 
especially in chronic cases, may remain asymptomatic. Symptoms, when present, 
typically emerge between two weeks to six months after infection. Symptoms of 
acute hepatitis include fatigue, fever, muscle and joint pain, nausea, vomiting, 
loss of appetite, weight loss, abdominal pain, jaundice, pale stool, dark urine, 
itchy skin, and an overall feeling of fatigue and malaise. Symptoms of chronic 
hepatitis may include blood in the stool or vomit and swelling of the lower extremities. 
Additionally, the skin may become yellow, and the whites of the eyes may also turn yellow.
The prevalence of Hepatitis B is widespread globally, as highlighted in 
\cite{libbus2009public,williams2006global}. Despite its global reach, there 
exists an effective vaccination to prevent infection, as emphasized 
in \cite{maynard1989global,shepard2006hepatitis}.

Mathematical epidemiology is a specialized field dedicated to exploring 
the dynamics of disease transmission. The dynamics of hepatitis viruses 
have been extensively studied through mathematical models, as demonstrated by numerous examples
\cite{wang2010global,safi2011effect,rachah2015mathematical,rodrigues2014vaccination,rodrigues2014cost}.
In a seminal study, Anderson et al. \cite{anderson1992may} presented a simple model to examine the 
influence of carriers on the transmission of hepatitis B. The transmission dynamics and control of the 
hepatitis B virus were modeled in \cite{zou2010modeling}, while a model forecasting a mechanism for 
eliminating hepatitis B was proposed in \cite{williams1996transmission}. Similar concepts 
were explored in \cite{medley2001hepatitis}.
Control analysis using an SIR epidemic model was suggested in 
\cite{zhao2000mathematical}. Further investigation of epidemic models 
with control strategies was conducted in \cite{bakare2014optimal,kamyad2014mathematical}. 
A model studying multiple endemic solutions was developed in \cite{onyango2014multiple}. 
Similarly, the dynamics of hepatitis B were explored in \cite{zhang2015modeling}.
Recent research, exemplified by works such as 
\cite{khan2016classification,khan2017transmission,nana2017hepatitis,debbouche2}, 
has developed epidemiological models to investigate the influences of various parameters 
on disease transmission and to propose control measures for infection elimination.

In these investigations, differential equations with integer orders were initially employed. 
However, it became evident over time that these models were insufficient for comprehending 
complex biological systems. Consequently, there has been a shift towards mathematical models 
with fractional orders, a trend that has recently gained prominence. The utilization of 
fractional derivatives and fractional integrals has found numerous applications 
in applied sciences and engineering. Many classical models have demonstrated limited 
accuracy in predicting the future dynamics of a system. In contrast, models incorporating 
fractional orders have proven more effective in capturing and retaining missing 
information \cite{abbasc,abbmma,atangana2020modelling,debbouche1,debbouche}. 
It is worth noting that classical derivatives may not adequately capture the dynamics 
between two distinct points \cite{samko1993fractional,baleanu2010new}.

Recent efforts have been directed towards enhancing discrete fractional calculus, 
as demonstrated in \cite{almatroud2023variable}. These advancements underscore the 
growing importance of discrete fractional calculus, as evidenced by works such as 
\cite{khan2021stability,baleanu2010new,atici2007transform}. Variable-order calculus 
is recognized as a natural extension of classical calculus \cite{bushnaq2023existence}, 
with foundational work in this area dating back to 1993 by Samko and his co-authors 
\cite{doi:10.1080/10652469308819027}. Subsequently, variable-order problems have 
found applications in fields such as photoelasticity \cite{soon2005variable}, 
and the stability and convergence of novel explicit finite-difference approaches 
for variable-order nonlinear fractional diffusion equations have been studied 
\cite{lin2009stability}. Various numerical schemes for variable-order problems 
have been developed, as referenced in \cite{bushnaq2022computation} 
and \cite{shah2022spectral}. Moreover, existence theories for variable-order problems 
have been established \cite{xu2013existence,razminia2012solution}. Given that 
variable-order operators extend classical ordinary and fractional orders, their 
utilization provides sophisticated tools for studying the dynamical systems 
of infectious diseases \cite{bushnaq2023existence}.

Building upon the aforementioned motivation and work, here we introduce
a mathematical model for the dynamics of the hepatitis B virus, 
employing fractional variable-order derivatives. The model encompasses 
populations of susceptible ($S(t)$), acute ($A(t)$), 
immune ($M(t)$),
chronic ($C(t)$), 
recovered ($R(t)$), and
vaccinated ($V(t)$), while
incorporating discrete fractional variable-order time derivatives. 
It is expressed as follows:
\begin{eqnarray}
\label{e1}
\left\{
\begin{array}{llllll}
\Delta^{\alpha(t)}S(t) & = &\delta d(1-rC)+wV-(\mu_{1}+\beta M+n\beta C+k)S, \\
\Delta^{\alpha(t)}A(t) & = &(\beta M+n\beta C)S-(\mu_{1}+\gamma)A, \\
\Delta^{\alpha(t)}M(t) & = &\gamma A+(\delta dr-\mu_{1}-\mu_{2}-a-m )M, \\
\Delta^{\alpha(t)}C(t) & = &aM -eC-\mu_{1}C, \\
\Delta^{\alpha(t)}R(t) & = &mM +eC-\mu_{1}R, \\
\Delta^{\alpha(t)}V(t) & =& \delta(1-d)+kS-\mu_{1}V-wV , \\
\end{array} \right.
\end{eqnarray}
where the initial conditions are given as 
$S(0)=S_{0}$, $A(0)=A_{0}$, $M(0)=M_{0}$, $C(0)=C_{0}$,
 $R(0)=R_{0}$ and $V(0)=V_{0}$, and where we use the fractional-order 	
$\alpha(t)\in(0,1)$. The unknown variables  
$(S(t),A(t),M(t),C(t),R(t),V(t))$ are in $ \mathbb{R}_{+}^{6}$. 
The total population is given by
$B(t)=S(t)+A(t)+M(t)+C(t)+R(t)+V(t)$. All the parameters 
$(\delta ,\mu_{1} ,\mu_{2} ,\beta ,\gamma ,a ,e,k,d,r,w,n,m )$ 
are positive real numbers and their values are provided in Table~\ref{table1}. 
Here, the delta variable-order difference of model \eqref{e1} 
is given in the sense of Caputo, where $\alpha(t)\in(0,1)$.

The central concept of the presented model revolves around the integration of discrete 
fractional variable-order time derivatives, representing a novel class of fractional 
derivatives with wide-ranging real-world applications. The proposed 
model explores the impact of different phases of infected individuals and various 
transmission routes on the dynamics of a hepatitis B virus model utilizing discrete 
time-fractional derivatives. This model provides significant advantages in 
comprehending the transmission dynamics of hepatitis B virus within the human population.
\begin{table}
\caption{Parameter values for our model \eqref{e1}, borrowed from \cite{yavuznew}.\label{table1}}
\[ 
\begin{array}{lccc}\\
\hline
\textrm	{Parameters} &  &\textrm {Parameter}\textrm	{Value}  \\ \hline
\delta  & \textrm {Birth Rate}	 &  0.0121  \\ 
\mu_{1}    & \textrm {Natural mortality rate}	 & 0.000034857 \\
\mu_{2}    & \textrm {Hepatitis-B related mortality rate}	 & 0.1019\\
\beta  & \textrm {Transmission coefficient of the disease}	 &  0.00014334    \\
\gamma   & \textrm {Transition rate from Latent population to Acute population} & 0.1989 \\
a   & \textrm { Transition rate of individuals with Acute infection
to carrier-class}	 &0.3387   \\
e   & \textrm {Recovery rate of individuals in the carrier class}	 &  0.0741   \\
k   &\textrm {Vaccination rate}	 & 0.8569\\
d  & \textrm {Rate of births without successful vaccination}	 & 0.00043102\\
r   & \textrm {Infected rate of mothers with HB Acute virus}	 & 0.0137  \\
w     & \textrm {The rate of decrease in immunity with the effect of vaccine}	 &0.9472\\
n    &  \textrm {Reduced transmission rate compared to Acute}	 & 0.7534\\
m  &  \textrm {Recovery rate of individuals with Acute infection}	 &  0.0277\\ \hline
\end {array} \]
\end{table} 

The paper is structured as follows. 
In Section~\ref{sec2} we recall necessary notions and results from the literature
needed in the sequel. In Section~\ref{sec:3}, we delve into the conditions 
for the existence and uniqueness of solutions as well as their boundedness and 
positivity. Section~\ref{sec:4} establishes the stability of the equilibrium points
while in Section~\ref{sec:5} we present some numerical results. 
We end the paper with Section~\ref{sec:6} of conclusion.

% -----------------------------------------------------

\section{Mathematical preliminaries}
\label{sec2}

Here, we introduce some definitions and notations from the papers 
\cite{huang2021discrete,ortigueira2019variable,abdeljawad2019variable,amarmat}. 
We denote $\mathbb{N}_a$ and $\mathbb{N}^T_a$ as  
$\mathbb{N}_a =\{a,a+1,a+2,\ldots\}$, $\mathbb{N}^T_a =\{a,a+1,a+2,\ldots,T\}$.

\begin{Def}
\label{Definition 2.1}  
Let $\alpha(t)>0$ and $\sigma(s)=s+1$. For $ u(t) $ defined on $ \mathbb{N}_a $, the
delta variable-order fractional sum of order $\alpha(t)$ is defined by 	
\begin{eqnarray}
\Delta^{-\alpha(t)}_{a}u(t)=\frac{1}{\Gamma(\alpha(t))}
\sum_{s=a}^{t-\alpha(t)}(t-\sigma(s))^{(\alpha(t)-1)}u(s),\label{2.1}  
\end{eqnarray}
where $t^{(\alpha(t))}$ is the discrete factorial functional given by 
$t^{(\alpha(t))}=\frac{\Gamma(t+1)}{\Gamma(t-\alpha(t)+1)}$.
\end{Def}

\begin{Def}
\label{Definition 2.1.1}
For $u(t)$ defined on $ \mathbb{N}_a, \alpha(t)>0, \alpha\notin\mathbb{N}$, 
the delta Caputo variable-order fractional difference is defined by
\begin{eqnarray}
^C \Delta^{\alpha(t)}_{a}u(t)=\Delta^{-(m-\alpha(t))}_{a}\Delta^{m}u(t)
=\frac{1}{\Gamma(m-\alpha(t))}\sum_{s=a}^{t
-(m-\alpha(t))}(t-\sigma(s))^{(m-\alpha(t)-1)}\Delta^{m}u(s),
\label{2.1.1}  
\end{eqnarray}
where $t\in \mathbb{N}_{a+m-\alpha(t)}$, $m=[\alpha(t)]+1$.
Note that the forward difference operator is defined by $\Delta u(t)=u(t+1)-u(t)$.
\end{Def}

\begin{thm}[See \cite{ghaziani2016stability,petravs2011fractional}]
Consider the following fractional variable-order discrete system:
\begin{equation}
\label{p1}
\Delta^{\alpha(t)}x =f(x), \quad x(0)=x_0  
\end{equation}
with $ x \in \mathbb{R}^{n} $ and $ \underline{\alpha}=\inf \alpha(t)$, 
$\overline{\alpha}=\sup \alpha(t)$, 
$0 < \underline{\alpha}  < \alpha(t) < \overline{\alpha}< 1$.
The equilibrium points of the system \eqref{p1} are solutions to the equation 
$ f(x)= 0$. An equilibrium is locally asymptotically stable if all 
the eigenvalues $ \lambda_i $ $(i=1,2 ,3,\ldots,n ) $ 
of the Jacobian matrix $ J=\Delta f $ evaluated at the equilibrium satisfy
\begin{equation}
\vert \arg(\lambda_i)  \vert <\dfrac{\pi}{2} \underline{\alpha}.
\end{equation}
On the other hand, if $ \vert \arg(\lambda_i)  \vert >\dfrac{\pi}{2} \overline{\alpha}$, 
then the equilibrium point is unstable. 
\end{thm}

\begin{thm}[See \cite{almatroud2023variable}]
Let $s\in \mathbb{N}_{a+1}$. The following holds: 
\begin{equation}
\sum_{k=a+1}^{s}(s-k+1)^{\alpha(s)-1}
=\dfrac{(s-a)^{\alpha(s)}}{\alpha(s)}.
\end{equation}
\end{thm}

% -------------------------------

\section{Wellposedness of the discrete time fractional model}
\label{sec:3}

In this section, we prove the existence and uniqueness 
of solutions for the given system \eqref{e1} 
(Theorem~\ref{thm:sol:exst:uniq}). Moreover,
we show that the solution is positive 
and bounded (Theorem~\ref{thm:pos}).

% ----------

\subsection{Existence and uniqueness of solution}   

Using the properties of fractional discrete variable-order calculus, 
we establish the existence and uniqueness of a solution to 
system \eqref{e1}. Now, system \eqref{e1} can be rewritten as,
\begin{eqnarray}
\label{e2}
\left\{
\begin{array}{llllll}
S(t) & = &S(0)+\Delta^{-\alpha(t)}( \delta d(1-rC)+wV-(\mu_{1}+\beta M+n\beta C+k)S), \\
A(t)&=&A(0)+\Delta^{-\alpha(t)}(\beta M+n\beta C)S-(\mu_{1}+\gamma)A, \\
M(t)&=&M(0)+\Delta^{-\alpha(t)} (\gamma A+(\delta dr-\mu_{1}-\mu_{2}-a-m)M), \\
C(t) & = & C(0)+\Delta^{-\alpha(t)}(aM -eC-\mu_{1}C), \\
R(t) & = &R(0)+\Delta^{-\alpha(t)}(mM +eC-\mu_{1}R), \\
V(t) & =&V(0)+\Delta^{-\alpha(t)} (\delta(1-d)+kS-\mu_{1}V-wV ).
\end{array} \right.
\end{eqnarray}
Using the definitions of variable-order calculus, it is easy to obtain 
\begin{eqnarray}
\label{e3}
\begin{array}{llllll}
S(t) & = &S(0)+\dfrac{1}{\Gamma(\alpha(t))}\sum_{s=0}^{t-\alpha(t)} 
(t-\sigma(s))^{(\alpha(t)-1)}( \delta d(1-rC)+wV-(\mu_{1}+\beta M+n\beta C+k)S), \\
A(t)&=&A(0)+\dfrac{1}{\Gamma(\alpha(t))}\sum_{s=0}^{t-\alpha(t)} 
(t-\sigma(s))^{(\alpha(t)-1)}(\beta M+n\beta C)S-(\mu_{1}+\gamma)A, \\
M(t)&=&M(0)+\dfrac{1}{\Gamma(\alpha(t))}\sum_{s=0}^{t-\alpha(t)} 
(t-\sigma(s))^{(\alpha(t)-1)} (\gamma A+(\delta dr-\mu_{1}-\mu_{2}-a-m )M), \\
C(t) & = & C(0)+\dfrac{1}{\Gamma(\alpha(t))}\sum_{s=0}^{t-\alpha(t)} 
(t-\sigma(s))^{(\alpha(t)-1)}(aM -eC-\mu_{1}C),\\
R(t) & = &R(0)+\dfrac{1}{\Gamma(\alpha(t))}\sum_{s=0}^{t-\alpha(t)} 
(t-\sigma(s))^{(\alpha(t)-1)}(mM +eC-\mu_{1}R), \\
V(t) & =&V(0)+\dfrac{1}{\Gamma(\alpha(t))}\sum_{s=0}^{t-\alpha(t)} 
(t-\sigma(s))^{(\alpha(t)-1)} (\delta(1-dC)+kS-\mu_{1}V-wV ).
\end{array}
\end{eqnarray}
We define the kernels $ L_{1}$ ,$L_{2}$, $L_{3}$, $L_{4}$, $L_{5}$ and $ L_{6} $ 
from the RHS of \eqref{e1}: 
\begin{eqnarray}
\label{e4}
\left\{
\begin{array}{llllll}
L_{1}(t,S(t))&=&\delta d(1-rC)+wV-(\mu_{1}+\beta M+n\beta C+k)S, \\
L_{2}(t,A(t))&=&(\beta M+n\beta C)S-(\mu_{1}+\gamma)A, \\
L_{3}(t,M(t))&=&\gamma A+(\delta dr-\mu_{1}-\mu_{2}-a-m )M, \\
L_{4}(t,C(t)) & = &aM -eC-\mu_{1}C, \\
L_{5}(t,R(t)) & = &mM +eC-\mu_{1}R, \\
L_{6}(t,V(t))&=&\delta(1-dC)+kS-\mu_{1}V-wV.
\end{array} \right.
\end{eqnarray}

\begin{Lemma}
\label{l1}
The kernels $ L_{1}$, $L_{2}$, $L_{3}$, 
$L_{4}$, $L_{5}$ and $ L_{6}$ 
satisfy a Lipschitz condition. 
\end{Lemma} 

\begin{proof}
Suppose $ S(t)$ and $S^{*}(t) $ to be the two unknowns 
of the first equation of \eqref{e4}. Then we have
\begin{eqnarray*}
\Vert L_{1}(t,S(t))-L_{1}(t,S^{*}(t))\Vert 
&=&\Vert(\mu_{1}+\beta M+n\beta C+k)(S(t)-S^{*}(t))\Vert.
\end{eqnarray*} 
Assume that 
$ m_{1} =\Vert\mu_{1}+\beta M+n\beta C+k \Vert$. Then the above can be rewritten as 
\begin{eqnarray*}
\Vert L_{1}(t,S(t))-L_{1}(t,S^{*}(t))\Vert & \leq &m_{1}\Vert S(t)-S^{*}(t)\Vert.
\end{eqnarray*} 
Using the similar procedure for $L_2, L_3, L_4, L_5$ and $L_6$, we get
\begin{eqnarray*}
\Vert L_{2}(t,A(t))-L_{2}(t,A^{*}(t))\Vert & \leq &m_{2}\Vert A(t)-A^{*}(t)\Vert,\nonumber \\
\Vert L_{3}(t,M(t))-L_{3}(t,M^{*}(t))\Vert & \leq &m_{3}\Vert M(t)-M^{*}(t)\Vert,\nonumber \\
\Vert L_{4}(t,C(t))-L_{4}(t,C^{*}(t))\Vert &  \leq &m_{4}\Vert C(t)-C^{*}(t)\Vert,\nonumber  \\
\Vert L_{5}(t,R(t))-L_{5}(t,R^{*}(t))\Vert & \leq &m_{5}\Vert R(t)-R^{*}(t)\Vert,\nonumber \\
\Vert L_{6}(t,V(t))-L_{6}(t,V^{*}(t))\Vert & \leq &m_{6}\Vert V(t)-V^{*}(t)\Vert,
\end{eqnarray*}
where $m_{1}$, $m_{2}$, $m_{3}$, $m_{4}$, $m_{5}$ and $ m_{6} $ are the Lipschitz constants. 
\end{proof}

Now using \eqref{e4} in \eqref{e3}, we get
\begin{eqnarray}
\label{e5}
\left\{
\begin{array}{llllll}
S(t) & = &S(0)+\dfrac{1}{\Gamma(\alpha(t))}\sum_{s=0}^{t-\alpha(t)} (t-\sigma(s))^{(\alpha(t)-1)}L_{1}(s,S(s)), \\
A(t)&=&A(0)+\dfrac{1}{\Gamma(\alpha(t))}\sum_{s=0}^{t-\alpha(t)} (t-\sigma(s))^{(\alpha(t)-1)}L_{2}(s,A(s)), \\
M(t)&=&M(0)+\dfrac{1}{\Gamma(\alpha(t))}\sum_{s=0}^{t-\alpha(t)} (t-\sigma(s))^{(\alpha(t)-1)} L_{3}(s,M(s)), \\
C(t) & = & C(0)+\dfrac{1}{\Gamma(\alpha(t))}\sum_{s=0}^{t-\alpha(t)} (t-\sigma(s))^{(\alpha(t)-1)}L_{4}(s,C(s)),\\
R(t) & = &R(0)+\dfrac{1}{\Gamma(\alpha(t))}\sum_{s=0}^{t-\alpha(t)} (t-\sigma(s))^{(\alpha(t)-1)} L_{5}(s,R(s)), \\
V(t) & =&V(0)+\dfrac{1}{\Gamma(\alpha(t))}\sum_{s=0}^{t-\alpha(t)} (t-\sigma(s))^{(\alpha(t)-1)}  L_{6}(s,V(s)).
\end{array} \right.
 \end{eqnarray} 
Then, using the recursive formula, we get
\begin{eqnarray}
\label{e6}
\left\{
\begin{array}{llllll}
S_{n}(t) & = &S(0)+\dfrac{1}{\Gamma(\alpha(t))}\sum_{s=0}^{t-\alpha(t)} (t-\sigma(s))^{(\alpha(t)-1)}L_{1}(s,S_{n-1}(s)), \\
A_{n}(t)&=&A(0)+\dfrac{1}{\Gamma(\alpha(t))}\sum_{s=0}^{t-\alpha(t)} (t-\sigma(s))^{(\alpha(t)-1)}L_{2}(s,A_{n-1}(s)), \\
M_{n}(t)&=&M(0)+\dfrac{1}{\Gamma(\alpha(t))}\sum_{s=0}^{t-\alpha(t)} (t-\sigma(s))^{(\alpha(t)-1)} L_{3}(s,M_{n-1}(s)), \\
C_{n}(t) & = & C(0)+\dfrac{1}{\Gamma(\alpha(t))}\sum_{s=0}^{t-\alpha(t)} (t-\sigma(s))^{(\alpha(t)-1)}L_{4}(s,C_{n-1}(s)),\\
R_{n}(t) & = &R(0)+\dfrac{1}{\Gamma(\alpha(t))}\sum_{s=0}^{t-\alpha(t)} (t-\sigma(s))^{(\alpha(t)-1)} L_{5}(s,R_{n-1}(s)), \\
V_{n}(t) & =&V(0)+\dfrac{1}{\Gamma(\alpha(t))}\sum_{s=0}^{t-\alpha(t)} (t-\sigma(s))^{(\alpha(t)-1)}  L_{6}(s,V_{n-1}(s)).
\end{array} \right.
\end{eqnarray} 
Subtract the \eqref{e6} with previous successive term and denote the resulting expressions as follows:  
$$
\Phi_{S,n}(t) =S_{n}(t)-S_{n-1}(t),\Phi_{A,n}(t) 
=A_{n}(t)-A_{n-1}(t), \Phi_{M,n}(t) =M_{n}(t)-M_{n-1}(t), 
$$ 
$$ 
\Phi_{R,n}(t) =R_{n}(t)-R_{n-1}(t), \Phi_{V,n}(t) =V_{n}(t)-V_{n-1}(t).
$$ 
Then, we have
\begin{eqnarray}
\Phi_{S,n}(t) & =&	S_{n}(t)-S_{n-1}(t)  
= \dfrac{1}{\Gamma(\alpha(t))}\sum_{s=0}^{t-\alpha(t)} (t-\sigma(s))^{(\alpha(t)-1)}(L_{1}(s,S_{n-1}(s))-L_{1}(s,S_{n-2}(s))),\nonumber \\
\Phi_{A,n}(t) & =&	A_{n}(t)-A_{n-1}(t)  
= \dfrac{1}{\Gamma(\alpha(t))}\sum_{s=0}^{t-\alpha(t)} (t-\sigma(s))^{(\alpha(t)-1)}(L_{2}(s,A_{n-1}(s))-L_{2}(s,A_{n-2}(s))),\nonumber \\
\Phi_{M,n}(t) & =&	M_{n}(t)-M_{n-1}(t)  
= \dfrac{1}{\Gamma(\alpha(t))}\sum_{s=0}^{t-\alpha(t)} (t-\sigma(s))^{(\alpha(t)-1)}(L_{3}(s,M_{n-1}(s))-L_{3}(s,M_{n-2}(s))),\nonumber \\
\Phi_{C,n}(t) & =&	C_{n}(t)-C_{n-1}(t)  
= \dfrac{1}{\Gamma(\alpha(t))}\sum_{s=0}^{t-\alpha(t)} (t-\sigma(s))^{(\alpha(t)-1)}(L_{4}(s,C_{n-1}(s))-L_{4}(s,C_{n-2}(s))),\nonumber \\
\Phi_{R,n}(t) & =&	R_{n}(t)-R_{n-1}(t)  
= \dfrac{1}{\Gamma(\alpha(t))}\sum_{s=0}^{t-\alpha(t)} (t-\sigma(s))^{(\alpha(t)-1)}(L_{5}(s,R_{n-1}(s))-L_{5}(s,R_{n-2}(s))),\nonumber \\
\Phi_{V,n}(t) & =&	V_{n}(t)-V_{n-1}(t)  
= \dfrac{1}{\Gamma(\alpha(t))}\sum_{s=0}^{t-\alpha(t)} (t-\sigma(s))^{(\alpha(t)-1)}(L_{6}(s,V_{n-1}(s))-L_{6}(s,V_{n-2}(s))).\nonumber\\
\end{eqnarray}
In the above,  
\begin{eqnarray}
\label{e7}
\left\{
\begin{array}{llllll}
S_{n}(t) =\sum_{j=0}^{n} \Phi_{S,j}(t) , \quad
A_{n}(t) =\sum_{j=0}^{n} \Phi_{A,j}(t) , \\[3mm]
M_{n}(t)=\sum_{j=0}^{n} \Phi_{M,j}(t) ,\quad
C_{n}(t) =\sum_{j=0}^{n} \Phi_{C,j}(t) , \\[3mm]
R_{n}(t) =\sum_{j=0}^{n} \Phi_{R,j}(t) ,\quad
V_{n}(t)=\sum_{j=0}^{n} \Phi_{V,j}(t).
\end{array} \right.
\end{eqnarray} 
Suppose 
\begin{eqnarray}
\label{e8}
\left\{
\begin{array}{llllll}
\Phi_{S,n-1}(t)=	S_{n-1}(t)-S_{n-2}(t) ,\quad
\Phi_{A,n-1}(t)=	A_{n-1}(t)-A_{n-2}(t)  , \\
\Phi_{M,n-1}(t) =	M_{n-1}(t)-M_{n-2}(t)  ,\quad
\Phi_{C,n-1}(t) =	C_{n-1}(t)-C_{n-2}(t)  , \\
\Phi_{R,n-1}(t)=	R_{n-1}(t)-R_{n-2}(t) ,\quad
\Phi_{V,n-1}(t) =	V_{n-1}(t)-V_{n-2}(t) .\\
\end{array} \right.
\end{eqnarray}
Thus, we obtain 
\begin{eqnarray}
\label{e9}
\left\{
\begin{array}{llllll}
\Vert\Phi_{S,n}(t)\Vert 
& < &	 \dfrac{m_{1}}{\Gamma(\alpha(t))}\sum_{s=0}^{t-\alpha(t)} (t-\sigma(s))^{(\alpha(t)-1)}\Vert\Phi_{S,n-1}(s)\Vert, \\[3mm]
\Vert\Phi_{A,n}(t)\Vert 
& < &	 \dfrac{m_{2}}{\Gamma(\alpha(t))}\sum_{s=0}^{t-\alpha(t)} (t-\sigma(s))^{(\alpha(t)-1)}\Vert\Phi_{A,n-1}(s)\Vert, \\[3mm]
\Vert\Phi_{M,n}(t)\Vert 
& < &	 \dfrac{m_{3}}{\Gamma(\alpha(t))}\sum_{s=0}^{t-\alpha(t)} (t-\sigma(s))^{(\alpha(t)-1)}\Vert\Phi_{M,n-1}(s)\Vert, \\[3mm]
\Vert\Phi_{C,n}(t)\Vert 
& < &	 \dfrac{m_{4}}{\Gamma(\alpha(t))}\sum_{s=0}^{t-\alpha(t)} (t-\sigma(s))^{(\alpha(t)-1)}\Vert\Phi_{C,n-1}(s)\Vert, \\[3mm]
\Vert\Phi_{R,n}(t)\Vert 
& < &	 \dfrac{m_{5}}{\Gamma(\alpha(t))}\sum_{s=0}^{t-\alpha(t)} (t-\sigma(s))^{(\alpha(t)-1)}\Vert\Phi_{R,n-1}(s)\Vert, \\[3mm]
\Vert\Phi_{V,n}(t)\Vert 
& < &	 \dfrac{m_{6}}{\Gamma(\alpha(t))}\sum_{s=0}^{t-\alpha(t)} (t-\sigma(s))^{(\alpha(t)-1)}\Vert\Phi_{V,n-1}(s)\Vert.
\end{array} \right.
\end{eqnarray} 

\begin{thm}
\label{thm:sol:exst:uniq}	
The solution of \eqref{e1} exists for   
$t \in  [0,\mathrm{T}],  $      
if $\displaystyle \frac{km_{i}}{\Gamma(\alpha(t))}  < 1 $, 
$i= 1, \ldots, 6$. Furthermore, the solution is unique if  
$ \Vert \chi(t) \Vert\left(1-\dfrac{km_{i}}{\Gamma(\alpha(t))}\right)> 0$ 
holds for $i= 1, \ldots, 6$. 
\end{thm}  

\begin{proof}
Using the Lemma~\ref{l1} and \eqref{e9}, we get
\begin{eqnarray}
\label{e10}
\left\{
\begin{array}{llllll}
\Vert\Phi_{S,n}(t)\Vert & < &	 \Vert S_{0}(t)\Vert (\dfrac{km_{1}}{\Gamma(\alpha(t))})^{n} , \\
\Vert\Phi_{A,n}(t)\Vert & < &	 \Vert A_{0}(t)\Vert (\dfrac{km_{2}}{\Gamma(\alpha(t))})^{n} , \\
\Vert\Phi_{M,n}(t)\Vert & < &	 \Vert M_{0}(t)\Vert (\dfrac{km_{3}}{\Gamma(\alpha(t))})^{n} , \\
\Vert\Phi_{C,n}(t)\Vert & < &	 \Vert C_{0}(t)\Vert (\dfrac{km_{4}}{\Gamma(\alpha(t))})^{n} , \\
\Vert\Phi_{R,n}(t)\Vert & < &	 \Vert R_{0}(t)\Vert (\dfrac{km_{5}}{\Gamma(\alpha(t))})^{n} , \\
\Vert\Phi_{V,n}(t)\Vert & < &	 \Vert V_{0}(t)\Vert (\dfrac{km_{6}}{\Gamma(\alpha(t))})^{n}.
\end{array} \right.
\end{eqnarray}
As $ n \rightarrow \infty $, we have $ \Vert\Phi_{\cdot,n}(t)\Vert\longrightarrow 0$. 
This shows the existence of solutions to \eqref{e1}. 

Suppose there are two solutions for the system \eqref{e1} 
$(S_{}(t) $, $ A_{}(t) $, $ M_{}(t) $, $ C_{}(t) $, $ R_{}(t) $ , $ V_{}(t))$ 
and $(S_{2}(t) $, $ A_{2}(t) $, $ M_{2}(t)$, $ C_{2}(t) $, $ R_{2}(t) $ , $ V_{2}(t))$. Then,
\begin{eqnarray}
\label{e11}
\left\{
\begin{array}{llllll}
S(t)-S_{2}(t) & =&\dfrac{1}{\Gamma(\alpha(t))}\sum_{s=0}^{t-\alpha(t)} (t-\sigma(s))^{(\alpha(t)-1)}(L_{1}(s,S(s))-L_{1}(s,S_{2}(s))), \\
A(t)-A_{2}(t) & =& \dfrac{1}{\Gamma(\alpha(t))}\sum_{s=0}^{t-\alpha(t)} (t-\sigma(s))^{(\alpha(t)-1)}(L_{2}(s,A(s))-L_{2}(s,A_{2}(s))), \\
M(t)-M_{2}(t)  & =& \dfrac{1}{\Gamma(\alpha(t))}\sum_{s=0}^{t-\alpha(t)} (t-\sigma(s))^{(\alpha(t)-1)}(L_{3}(s,M_(s))-L_{3}(s,M_{2}(s))), \\
C(t)-C_{2}(t)  & =& \dfrac{1}{\Gamma(\alpha(t))}\sum_{s=0}^{t-\alpha(t)} (t-\sigma(s))^{(\alpha(t)-1)}(L_{4}(s,C(s))-L_{4}(s,C_{2}(s))), \\
R(t)-R_{2}(t)  & =& \dfrac{1}{\Gamma(\alpha(t))}\sum_{s=0}^{t-\alpha(t)} (t-\sigma(s))^{(\alpha(t)-1)}(L_{5}(s,R(s))-L_{5}(s,R_{2}(s))), \\
V(t)-V_{2}(t)   & =& \dfrac{1}{\Gamma(\alpha(t))}\sum_{s=0}^{t-\alpha(t)} (t-\sigma(s))^{(\alpha(t)-1)}(L_{6}(s,V(s))-L_{6}(s,V_{2}(s))).
\end{array} \right.
\end{eqnarray} 
From the above, we get
\begin{eqnarray}
\label{e12}
\left\{
\begin{array}{llllll}
\Vert	S(t)-S_{2}(t) \Vert
& =&\Vert\dfrac{1}{\Gamma(\alpha(t))}\sum_{s=0}^{t-\alpha(t)} (t-\sigma(s))^{(\alpha(t)-1)}(L_{1}(s,S(s))-L_{1}(s,S_{2}(s)))\Vert, \\
\Vert	A(t)-A_{2}(t) \Vert 
& =& \Vert\dfrac{1}{\Gamma(\alpha(t))}\sum_{s=0}^{t-\alpha(t)} (t-\sigma(s))^{(\alpha(t)-1)}(L_{2}(s,A(s))-L_{2}(s,A_{2}(s)))\Vert, \\
\Vert M(t)-M_{2}(t) \Vert 
& =& \Vert\dfrac{1}{\Gamma(\alpha(t))}\sum_{s=0}^{t-\alpha(t)} (t-\sigma(s))^{(\alpha(t)-1)}(L_{3}(s,M_(s))-L_{3}(s,M_{2}(s)))\Vert, \\
\Vert	C(t)-C_{2}(t) \Vert  
& =&\Vert \dfrac{1}{\Gamma(\alpha(t))}\sum_{s=0}^{t-\alpha(t)} (t-\sigma(s))^{(\alpha(t)-1)}(L_{4}(s,C(s))-L_{4}(s,C_{2}(s)))\Vert, \\
\Vert R(t)-R_{2}(t) \Vert 
& =& \Vert\dfrac{1}{\Gamma(\alpha(t))}\sum_{s=0}^{t-\alpha(t)} (t-\sigma(s))^{(\alpha(t)-1)}(L_{5}(s,R(s))-L_{5}(s,R_{2}(s)))\Vert, \\
\Vert V(t)-V_{2}(t) \Vert  
& =& \Vert \dfrac{1}{\Gamma(\alpha(t))}\sum_{s=0}^{t-\alpha(t)} (t-\sigma(s))^{(\alpha(t)-1)}(L_{6}(s,V(s))-L_{6}(s,V_{2}(s)))\Vert.
\end{array} \right.
\end{eqnarray} 
The above leads to, 
\begin{eqnarray}
\label{e13}
\left\{
\begin{array}{llllll}
\Vert	S(t)-S_{2}(t) \Vert
& =&\dfrac{1}{\Gamma(\alpha(t))}\sum_{s=0}^{t-\alpha(t)} (t-\sigma(s))^{(\alpha(t)-1)}\Vert(L_{1}(s,S(s))-L_{1}(s,S_{2}(s)))\Vert,\nonumber \\
\Vert	A(t)-A_{2}(t) \Vert 
& =& \dfrac{1}{\Gamma(\alpha(t))}\sum_{s=0}^{t-\alpha(t)} (t-\sigma(s))^{(\alpha(t)-1)}\Vert(L_{2}(s,A(s))-L_{2}(s,A_{2}(s)))\Vert,\nonumber \\
\Vert M(t)-M_{2}(t) \Vert 
& =& \dfrac{1}{\Gamma(\alpha(t))}\sum_{s=0}^{t-\alpha(t)} (t-\sigma(s))^{(\alpha(t)-1)}\Vert(L_{3}(s,M_(s))-L_{3}(s,M_{2}(s)))\Vert,\nonumber \\
\Vert	C(t)-C_{2}(t) \Vert  
& =& \dfrac{1}{\Gamma(\alpha(t))}\sum_{s=0}^{t-\alpha(t)} (t-\sigma(s))^{(\alpha(t)-1)}\Vert(L_{4}(s,C(s))-L_{4}(s,C_{2}(s)))\Vert,\nonumber \\
\Vert R(t)-R_{2}(t) \Vert 
& =& \dfrac{1}{\Gamma(\alpha(t))}\sum_{s=0}^{t-\alpha(t)} (t-\sigma(s))^{(\alpha(t)-1)}\Vert(L_{5}(s,R(s))-L_{5}(s,R_{2}(s)))\Vert,\nonumber \\
\Vert V(t)-V_{2}(t) \Vert  
& =&  \dfrac{1}{\Gamma(\alpha(t))}\sum_{s=0}^{t-\alpha(t)} (t-\sigma(s))^{(\alpha(t)-1)}\Vert(L_{6}(s,V(s))-L_{6}(s,V_{2}(s)))\Vert.
\end{array} \right.
\end{eqnarray} 
Using Lemma~\ref{l1}, we have
\begin{eqnarray}
\label{e14}
\left\{
\begin{array}{llllll}
\Vert	S(t)-S_{2}(t) \Vert& \leq &\dfrac{m_{1}}{\Gamma(\alpha(t))}\sum_{s=0}^{t-\alpha(t)} (t-\sigma(s))^{(\alpha(t)-1)}\Vert	S(t)-S_{2}(t) \Vert, \\
\Vert	A(t)-A_{2}(t) \Vert & \leq & \dfrac{m_{2}}{\Gamma(\alpha(t))}\sum_{s=0}^{t-\alpha(t)} (t-\sigma(s))^{(\alpha(t)-1)}\Vert	A(t)-A_{2}(t) \Vert, \\
\Vert M(t)-M_{2}(t) \Vert & \leq & \dfrac{m_{3}}{\Gamma(\alpha(t))}\sum_{s=0}^{t-\alpha(t)} (t-\sigma(s))^{(\alpha(t)-1)}\Vert M(t)-M_{2}(t) \Vert, \\
\Vert	C(t)-C_{2}(t) \Vert  & \leq & \dfrac{m_{4}}{\Gamma(\alpha(t))}\sum_{s=0}^{t-\alpha(t)} (t-\sigma(s))^{(\alpha(t)-1)}	\Vert	C(t)-C_{2}(t) \Vert , \\
\Vert R(t)-R_{2}(t) \Vert & \leq & \dfrac{m_{5}}{\Gamma(\alpha(t))}\sum_{s=0}^{t-\alpha(t)} (t-\sigma(s))^{(\alpha(t)-1)}\Vert R(t)-R_{2}(t) \Vert, \\
\Vert V(t)-V_{2}(t) \Vert  & \leq &  \dfrac{m_{6}}{\Gamma(\alpha(t))}\sum_{s=0}^{t-\alpha(t)} (t-\sigma(s))^{(\alpha(t)-1)}\Vert V(t)-V_{2}(t) \Vert.
\end{array} \right.
\end{eqnarray} 
As a consequence, we get
\begin{eqnarray}
\label{e15}
\left\{
\begin{array}{llllll}
\Vert	S(t)-S_{2}(t) \Vert \left(1-\dfrac{km_{1}}{\Gamma(\alpha(t))}\right)& \leq &0, \\
\Vert	A(t)-A_{2}(t) \Vert \left(1-\dfrac{km_{2}}{\Gamma(\alpha(t))}\right) & \leq & 0, \\
\Vert M(t)-M_{2}(t) \Vert\left(1-\dfrac{km_{3}}{\Gamma(\alpha(t))}\right) & \leq & 0, \\
\Vert	C(t)-C_{2}(t) \Vert \left(1-\dfrac{km_{4}}{\Gamma(\alpha(t))}\right) & \leq & 0, \\
\Vert R(t)-R_{2}(t) \Vert \left(1-\dfrac{km_{5}}{\Gamma(\alpha(t))}\right)& \leq & 0, \\
\Vert V(t)-V_{2}(t) \Vert \left(1-\dfrac{km_{6}}{\Gamma(\alpha(t))}\right) & \leq &  0.
\end{array} \right.
\end{eqnarray}
The result is proved. 
\end{proof}

% --------------------

\subsection{Non-negativity and boundedness of the solutions}

Now we address the positivity and boundedness of the 
solutions of model \eqref{e1}. The main result here 
establishes that $S(t)$, $A(t)$, $M(t)$, $C(t)$, $R(t)$, 
and $V(t)$ are all positive and bounded.

\begin{rem}
\label{r1}
As $ (1-d) $ represents the efficacy of the successful vaccine 
and $\delta$ is a non-negative parameter, then
$ \delta(1-d) $ is always non-negative.
\end{rem}

\begin{thm}
\label{thm:pos}	
The solutions $S(t)$, $A(t)$, $M(t)$, $C(t)$, $R(t)$ 
and $ V(t)$ of \eqref{e1} are in $ \mathbb{R}_{+}^{6}$
for all $t$.
\end{thm}

\begin{proof}
Suppose a general fractional variable order discrete time model \eqref{e1} is given as
\begin{eqnarray}
\label{e16}
\left\{
\begin{array}{llllll}
\Delta^{\alpha(t)}S(t)\arrowvert_{S=0} & = &\delta d(1-rC)+wV \ge0, \\
\Delta^{\alpha(t)}A(t)\arrowvert_{A=0} & = &(\beta M+n\beta C)S \ge0, \\
\Delta^{\alpha(t)}M(t)\arrowvert_{M=0} & = &\gamma A \ge0, \\
\Delta^{\alpha(t)}C(t) \arrowvert_{C=0}& = &aM  \ge0, \\
\Delta^{\alpha(t)}R(t) \arrowvert_{R=0}& = &mM +eC \ge0, \\
\Delta^{\alpha(t)}V(t)\arrowvert_{V=0} & =& \delta(1-d)+kS \ge0.
\end{array} \right.
\end{eqnarray} 
Using Remark~\ref{r1}, we get that $S(t),A(t),M(t),C(t),R(t)$ 
and $ V(t) $  are positive.
\end{proof}

\begin{thm}
Let $ B(t) $ be the total population of \eqref{e1}. Then,
$$ B(t)={(S(t),A(t),M(t),C(t),R(t),V(t))}\in \mathbb{R}_{+}^{6}  $$ 
and 
$$ 
0< S(t)+A(t)+M(t)+C(t)+R(t)+V(t) <  \dfrac{\delta}{\mu_{1}}.  
$$	
\end{thm}

\begin{proof}
Here $ B(t) $ represents the  total population and 
it is given as $ B=S+A+M+C+R+V$. Then the discrete Caputo 
variable order fractional derivative of the total population is:
$$
\Delta^{\alpha(t)}B(t) =\delta-\mu_{1}B-\mu_{2}M. 
$$
Then, we have 
$$\Delta^{\alpha(t)}B(t) <  \delta-\mu_{1}B. $$
Using the Laplace transform, it is easy to obtain
$$
s^{\alpha(s)}\tilde{B}(s)-s^{\alpha(s)-1}\tilde{B}(0) 
< \dfrac{ \delta}{s}-\mu_{1}\tilde{B}(s).
$$
Then
$$
\tilde{B}(s) <  \dfrac{s^{\alpha(s)-1}}{(s^{\alpha(s)}+\mu_{1})}
\tilde{B}(0)+\dfrac{\delta s^{-1}}{(s^{\alpha(s)}+\mu_{1})}.
$$
Further, using the inverse Laplace transform, we have 
$$
B(s)<  \dfrac{\delta}{\mu_{1}}(\mu_{1}t^{\alpha(t)}E_{\alpha(t),
\alpha(t+1)}(-\mu_{1}t^{\alpha(t)}))
+E_{\alpha(t),1}(-\mu_{1}t^{\alpha(t)})< \dfrac{\delta}{\mu_{1}}.
$$
This completes the proof. 
\end{proof}

% ---------------------- 

\section{Stability analysis of the fractional hepatitis-B virus model}
\label{sec:4}
		 
The epidemiological model \eqref{e1} is analyzed for equilibrium points: 
the disease-free state denoted by $ B_{1} $ and endemic states denoted 
by $ B_{2} $. These equilibrium points can be obtained by setting
$$
\Delta^{\alpha(t)}S(t)=\Delta^{\alpha(t)}A(t)
= \Delta^{\alpha(t)}M(t)=\Delta^{\alpha(t)}C(t)
=\Delta^{\alpha(t)}R(t) =\Delta^{\alpha(t)}V(t)=0. 
$$ 
Using the above conditions in \eqref{e1}, we get
\begin{eqnarray}
\label{e17}
\left\{
\begin{array}{llllll}
\delta d(1-rC)+wV-(\mu_{1}+\beta M+n\beta C+k)S & = &0, \\
(\beta M+n\beta C)S-(\mu_{1}+\gamma)A& = &0, \\
\gamma A+(\delta dr-\mu_{1}-\mu_{2}-a-m)M& = &0, \\
aM -eC-\mu_{1}C& = &0, \\
mM +eC-\mu_{1}R& = &0, \\
\delta(1-d)+kS-\mu_{1}V-wV & = &0.
\end{array} \right.
\end{eqnarray}
We obtain the disease free equilibrium point as given below: 
$$
B_{1}=\left(\frac{\delta\left(d \mu_{1}+w\right)}{\mu_{1}
\left(\mu_{1}+k+w\right)}, 0,0,0,0, \frac{\delta\left(\mu_{1}
+k-\mu_{1} d\right)}{\mu_{1}\left(\mu_{1}+k+w\right)}\right).
$$
Further, a unique endemic equilibrium is given as 
$B_{2}=\left(S_{2}, A_{2}, M_{2}, C_{2}, R_{2}, V_{2}\right)$, where
\begin{equation}
\begin{array}{rl}
S_{2}&=\displaystyle\frac{\left(-\delta dr +\mu_{1}+\mu_{2}
+a+m\right)\left(\mu_{1}+\gamma\right)
\left(e+\mu_{1}\right)}{\gamma\beta\left(e+\mu_{1}+na\right)} ;\\[4mm]
V_{2}&=\displaystyle\frac{\delta-\delta d+k S_{2}}{\mu_{1}+w} ;\\[4mm]
C_2&=\displaystyle\frac{(\delta d
+wV_2-(\mu_{1}+k)S_2)a}{ar \delta d+S_2\beta(e+\mu_1+na)};\\[4mm]
A_2&=\displaystyle\frac{(\delta dr-\mu_1-\mu_2-a-m)(e+\mu_1)}{a\gamma}C_2;\\[4mm]
M_2&=\displaystyle\frac{e+\mu_1}{a}C_2;\\[4mm]
R_2&=\displaystyle\frac{m(e+\mu_1)+ea}{a\mu_1}C_2.
\end{array}
\end{equation}
To derive the expression for the reproduction number $ \mathcal R_{0} $ 
of the system \eqref{e1}, we employ the method outlined 
in \cite{yavuznew}. The necessary matrices are obtained as follows:
$$
\mathcal{F}=\left(\begin{array}{c}
0 \\
\left(\beta M+n\beta C\right) S \\
0 \\
0 \\
0 \\
0
\end{array}\right), 
\quad \mathcal{V}=\left(\begin{array}{c}
-\delta d(1-rC)+wV+Q_1 S \\
\left(\mu_{1}+\gamma \right) A \\
-\gamma A-\delta dr M+Q_2 M \\
-aM +eC+\mu_{1}C \\
- mM -eC+\mu_{1}R \\
V\left(\mu_{1}+w\right) -k S-C\left(1-d\right)
\end{array}\right),
$$
where $Q_1=\mu_{1}+\beta M+n\beta C+k $ and $Q_2=\mu_{1}+\mu_{2}+a+m $.
Now, for linearization, the Jacobian matrix of the upper matrix 
at the disease-free state is given by  
$$
\mathbf{F}=\left(\begin{array}{cccc} 
0 & \beta S_0 &n \beta S_0 &0\\
0 & 0 & 0 &0\\
0 & 0 & 0 &0\\0 & 0 & 0 &0
\end{array}\right), \quad \mathbf{V}
=\left(
\begin{array}{cccc}
\mu_{1}+\gamma  & 0 & 0  &0\\
-\gamma & \mu_{1}+\mu_{2}+a+m-\delta dr & 0  &0\\
0 & -a & e+ \mu_{1} &0\\
0&-m&-e&\mu_1
\end{array}\right) .
$$
So the spectral of the matrix $F V^{-1}$ 
is the expression of $\mathcal{R}_0$ and written as follows:
$$
\mathcal{R}_0=\frac{\delta \beta \gamma\left(\mu_{1}+e+n a\right)
\left(w+d \mu_{1}\right)}{\mu_{1}\left(\mu_{1}+\gamma\right)\left(\mu_{1}
+e\right)\left(w+k+\mu_{1}\right)\left(\mu_{1}+\mu_{2}+a+m-r \delta d)\right.}.
$$ 

We now investigate the local stability. 

\begin{thm}
The disease-free equilibrium $ B_{1} $ of the proposed discrete fractional 
variable-order hepatitis B model \eqref{e1} is locally 
asymptotically stable if $\mathcal{R}_0 < 1$. 
\end{thm}

\begin{proof}
The Jacobian of model at $B_{1}$ given by 
$$
\mathcal{J}\left(B_{1}\right)
=\left(\begin{array}{cccccc}
-\left(\mu_{1}+k\right) & 0 & -\beta S_0& -rd\delta-n \beta S_0 & 0 & w \\
0 & -\left(\mu_{1}+\gamma\right) & \beta S_0 & n \beta S_0 & 0 & 0 \\
0 & \gamma& \delta dr-\left(\mu_{1}+\mu_{2}+a+m\right) & 0 & 0 & 0 \\
0 & 0 & a & -\left(e+\mu_{1}\right) & 0 & 0 \\
0 & 0 & m & e & -\mu_{1} & 0 \\
k & 0 & 0 & 0 & 0 & -\left(\mu_{1}+w\right)
\end{array}\right) .
$$
The characteristic equation is
$$
(\lambda+\mu_1)(\lambda^2+a_1\lambda+a_2)(\lambda^3+b_1\lambda^2+b_2\lambda+b_3)=0,
$$
where
$$
a_1=2\mu_1+k+w;~~a_2=(\mu_1+k)(\mu_1+w)+kw;~~b_1=a+e+\gamma+m+3\mu_{1}-\mu_2-\delta dr;
$$
$$
b_2=-(\gamma \beta S_0+(\mu_1+\gamma)(\delta dr-Q_2)+(\delta dr-Q_2)(e+\mu_1)-(\mu_1+\gamma)(e+\mu_1));
$$
$$
b_3=-\gamma \beta S_0(na+e+\mu_1)+(\mu_1+\gamma)(\delta dr-Q_2)(e+\mu_1),
$$
here $Q_2=\left(\mu_{1}+\mu_{2}+a+m\right)$.
By the Routh--Hurwitz criterion, 
$  a_{1}>0 $, $a_{2}>0$, $b_{1}>0$, $b_{3}>0 $, $b_{1}b_{2} > b_{3}$, 
and if  $ \mathcal{R}_0 <1$, then  $ B_{1} $ is  asymptotically stable. 
\end{proof}

% ---------------------------------------

\section{Numerical results}
\label{sec:5}

This section provides a discussion on fractional Caputo's 
hepatitis B virus model \eqref{e1}. We utilize the numerical 
values presented in Table~\ref{table1} for our analysis. 
The considered hepatitis B virus model with a fractional variable order 
in discrete time is solved numerically utilizing the methodology outlined 
in \cite{bakare2014optimal} and \cite{gao2021stability}.

The initial conditions for the system \eqref{e1} are set as follows: 
$S(0) = 0.8565$, $A(0) = 0.01223$, $M(0) = 0.10675$, $C(0) = 0.00489$, 
$R(0) = 0.00734$, and $V(0) = 0.01224$, as in \cite{yavuznew}.
We consider different values for the fractional variable order, 
denoted as $\alpha(t)$:
$$
\alpha_{1}  =\left(
\begin{array}{ccc} 
0.95 & 0.9 &0.85 \\
\end{array}\right),  
\quad
\alpha_{2}  =
\left(\begin{array}{ccc}
0.93 & 0.84 &0.88 \\
\end{array}\right) ,  
\quad
\alpha_{3}  =\left(
\begin{array}{ccc}
0.78 & 0.8 &0.82 \\
\end{array}\right).
$$ 

\begin{figure}
\includegraphics[scale=0.8]{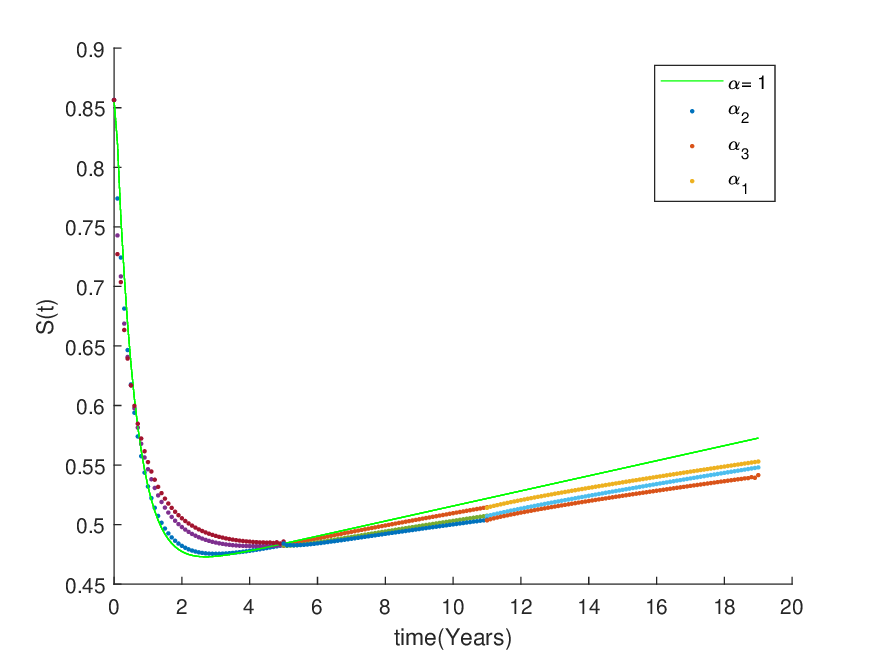}
\caption{\label{fig:s}
The graph illustrates the dynamics 
of the susceptible population $(S(t))$ with respect to the 
fractional variable-order parameter $\alpha(k)$.}
\end{figure}

\begin{figure}
\includegraphics[scale=0.8]{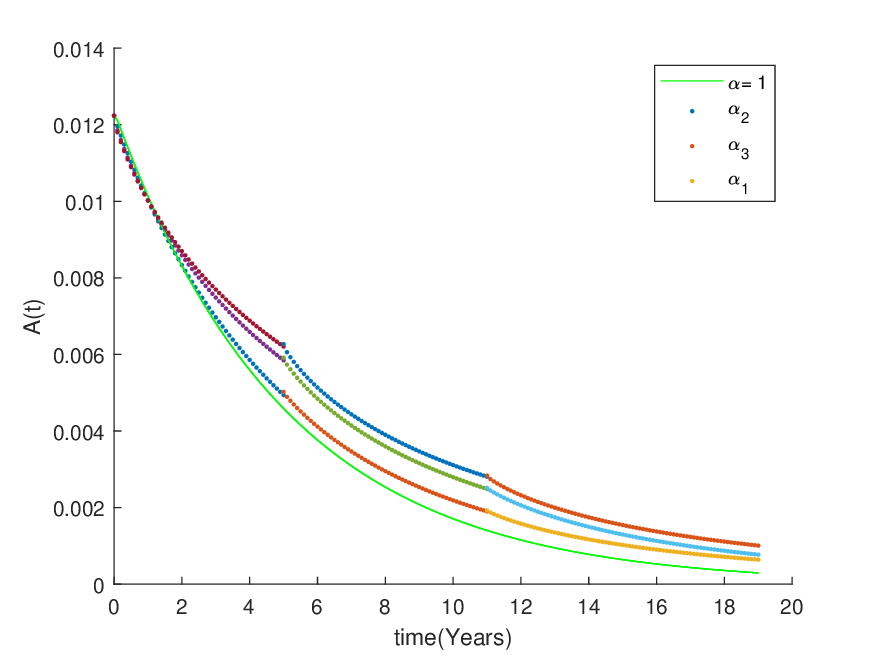}
\caption{\label{fig:a}
The graph illustrates the dynamics of the acute individual 
class $(A(t))$ with respect to the fractional variable-order parameter $\alpha(k)$.}
\end{figure}

\begin{figure}
\includegraphics[scale=0.8]{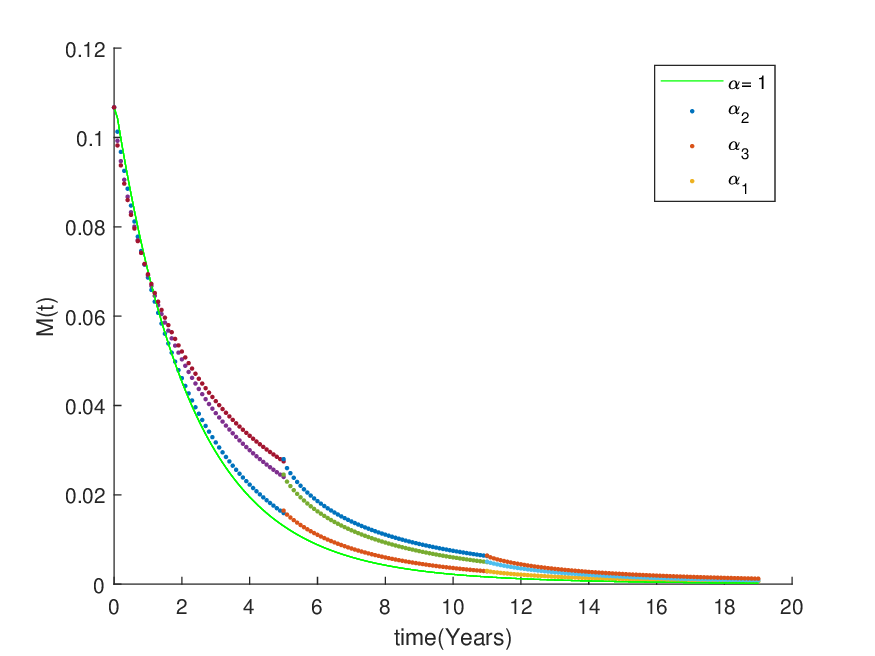}
\caption{\label{fig:m}
The graph illustrates the dynamics of the immune individuals 
class $(M(t))$ with respect to the fractional 
variable-order parameter $\alpha(k)$.}
\end{figure}

The results are presented in Figs.~\ref{fig:s}--\ref{fig:v},  
showcasing the dynamic behaviors of the susceptible, 
acute, immune, chronic, recovered, and vaccinated population groups. 

Fig.~\ref{fig:s} illustrates that as the value of $ \alpha $ increases, 
the number of susceptible individuals also increases, indicating 
a proportional relationship. Conversely, Figs.~\ref{fig:a} and \ref{fig:m} 
show an inverse proportionality between the fractional order and the number 
of acute and immune individuals, respectively. Notably, fractional and 
integer order cases tend to converge and stabilize over time. 
In Fig.~\ref{fig:c}, for fractional values, the peak in carrier 
individuals is observed in the year seven, whereas for 
$ \alpha=1 $, the peak occurs in earlier years. 
Fig.~\ref{fig:r} shows an increase in the value of $ \alpha $ 
after two years. Fig.~\ref{fig:v} reveals the dynamics of the 
vaccinated population, indicating a direct relationship between 
the fractional-order parameter $ \alpha(k) $ and the vaccinated population.

The computational results demonstrate that, 
by incorporating discrete time and variable fractional orders, 
the hepatitis B virus model is capable of capturing richer 
and more complex dynamics. Moreover, it offers valuable insights 
and enhances understanding of the model. Various morphological changes 
in the unknown variables, due to the influence of discrete time 
and the variable fractional order, are visible and comparable, 
as illustrated in Fig.~\ref{fig:s}--\ref{fig:v}.

\begin{figure}
\includegraphics[scale=0.8]{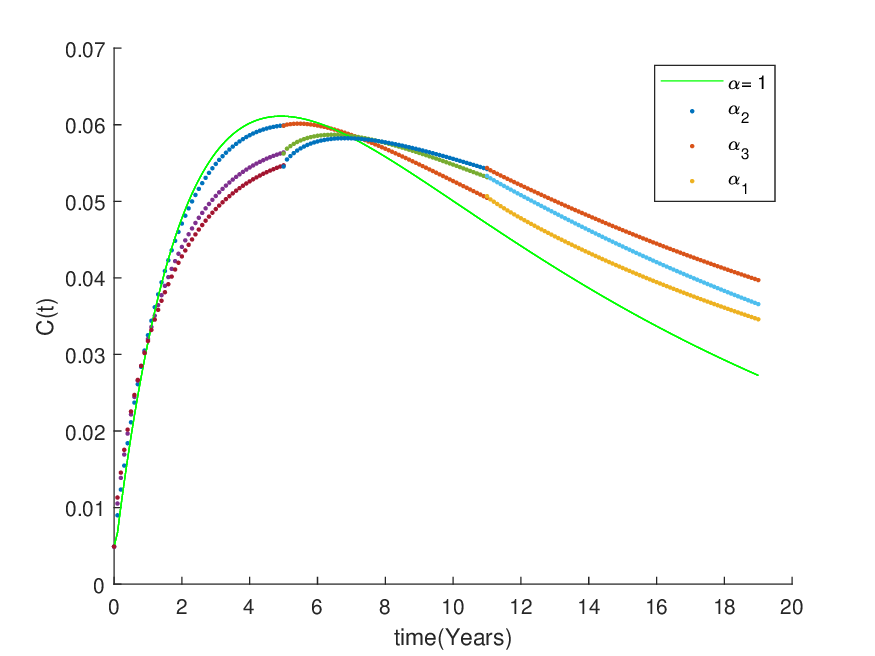}
\caption{\label{fig:c}
The graph illustrates the dynamics of the carrier (chronic) 
individual class $(C(t))$ with respect to the 
fractional variable-order parameter $\alpha(k)$.}
\end{figure}

\begin{figure}
\includegraphics[scale=0.8]{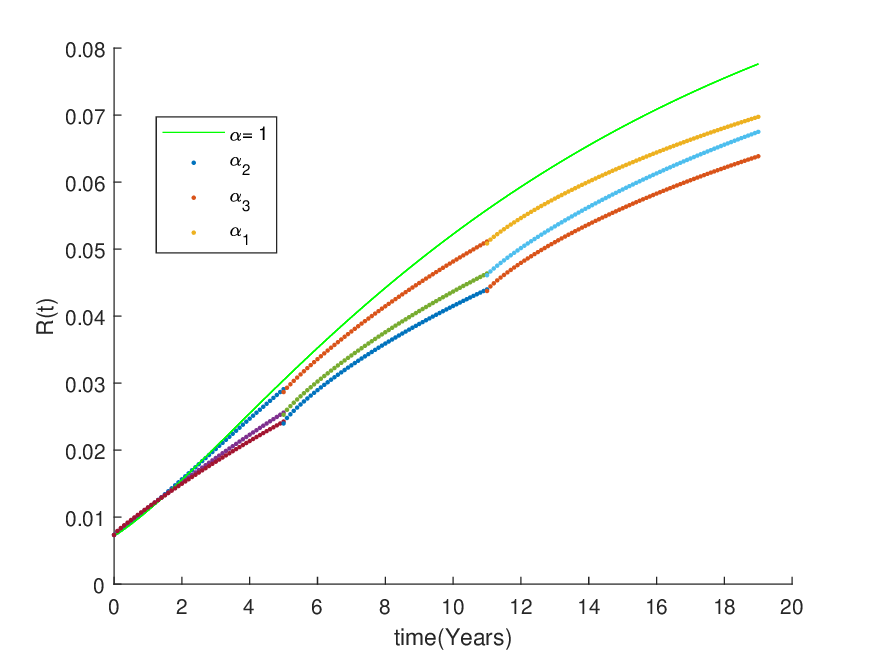}
\caption{\label{fig:r}
The graph illustrates the dynamics of the recovered individuals  
class $(R(t))$ with respect to the fractional variable-order parameter $\alpha(k)$.}
\end{figure}

\begin{figure}
\label{f6}	
\includegraphics[scale=0.8]{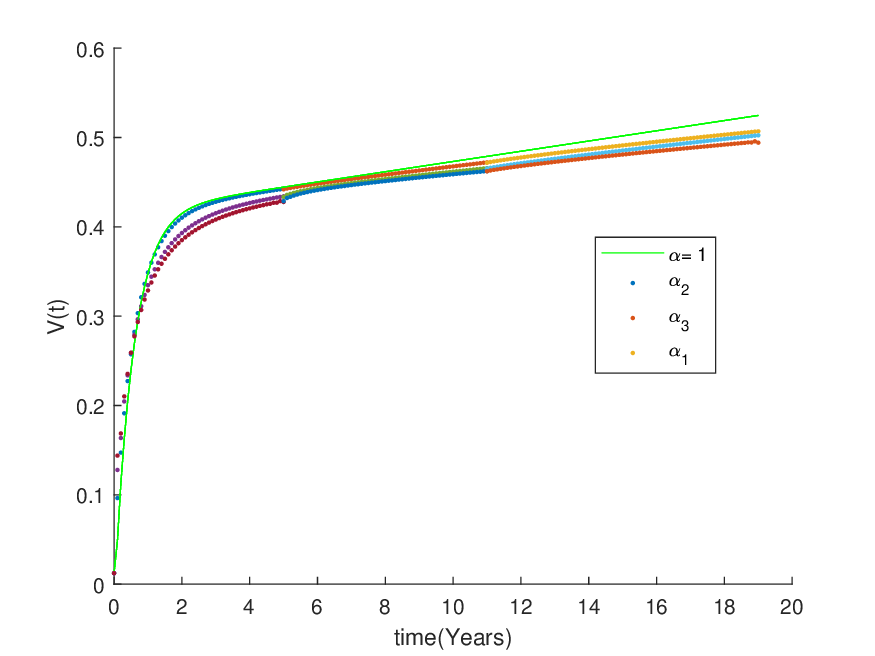}
\caption{\label{fig:v}
The graph illustrates the dynamics of the vaccination class $(V(t))$ 
with respect to the fractional variable-order parameter $\alpha(k)$.}
\end{figure}

% --------------------------------------

\section{Conclusion}
\label{sec:6}

In this study, we have established the existence and uniqueness 
of a solution for a novel discrete fractional variable-order model 
of the hepatitis-B virus. Initially, we ensured the positivity 
and boundedness of model's solutions. To further investigate the 
local stability of the proposed system, we computed the basic reproduction 
number $\mathcal{R}_0$. If $\mathcal{R}_0 < 1$, we have proven that 
the disease-free equilibrium $B_{1}$ is locally asymptotically stable.
To illustrate the dynamics of the model under various 
$\alpha(t)$ values, we have used 
Matlab to generate dynamic graphs for the proposed discrete 
fractional variable-order hepatitis-B virus model.
Our results show that, due to the incorporation 
of discrete time and variable fractional orders,
the new model is able to cover more rich 
and more complex dynamics when compared to 
models available in the literature,
providing valuable insights. 

% --------------------------------------

~\\
\textbf{CRediT authorship contribution statement}\\
M. Boukhobza: Writing - original draft, Software, Data curation, Resources; 
A. Debbouche: Writing - review $\&$ editing, Validation, Project administration,  Supervision; 
L. Shangerganesh: Conceptualization, Formal analysis, Visualization, Investigation;
D.F.M. Torres: Writing - review $\&$ editing, Methodology, Validation, Funding acquisition.\\
~\\
\textbf{Declaration of competing interest}\\
The authors declare that they have no known competing financial
interests or personal relationships that could have appeared to influence
the work reported in this paper.\\
~\\
\textbf{Acknowledgments}\\
The authors would like to express their sincere thanks 
and respect to the editors and the anonymous reviewers. This paper belongs 
to the first author PhD program under supervision of A. Debbouche 
in collaboration with 3rd and 4th co-authors.
Debbouche and Torres were supported by 
\emph{Funda\c{c}\~{a}o para a Ci\^{e}ncia e a Tecnologia} (FCT) 
within project number UIDB/04106/2020 (CIDMA).

% --------------------------------------

% --------------------------------------


\small

\begin{thebibliography}{10}
	
\bibitem{abbmma}
S.~Abbas, S.~Tyagi, P.~Kumar, V.~S. Ert\"urk, and S.~Momani.
\newblock Stability and bifurcation analysis of a fractional-order model of
cell-to-cell spread of hiv-1 with a discrete time-delay.
\newblock {\em Mathematical Methods in the Applied Sciences},
45(11):7081--7095, 2022.
	
\bibitem{abdeljawad2019variable}
T.~Abdeljawad, R.~Mert, and D.~F.~M. Torres.
\newblock Variable order Mittag-Leffler fractional operators on isolated time
scales and application to the calculus of variations.
\newblock {\em Fractional Derivatives with Mittag-Leffler Kernel: Trends and
Applications in Science and Engineering}, pages 35--47, 2019.
{\tt arXiv:1809.02029}
	
\bibitem{almatroud2023variable}
O.~A. Almatroud, A.~Hioual, A.~Ouannas, M.~M. Sawalha, S.~Alshammari, and M.~Alshammari.
\newblock On variable-order fractional discrete neural networks: Existence,
uniqueness and stability.
\newblock {\em Fractal and Fractional}, 7(2):118, 2023.
	
\bibitem{abbasc}
J.~Alzabut, S.~Tyagi, and S.~Abbas.
\newblock Discrete fractional-order bam neural networks with leakage delay:
Existence and stability results.
\newblock {\em Asian Journal of Control}, 22(1):143--155, 2018.
	
\bibitem{anderson1992may}
R.~M. Anderson and M.~Robert.
\newblock May. infectious diseases of humans: dynamics and control, 1992.
	
\bibitem{atangana2020modelling}
A.~Atangana.
\newblock Modelling the spread of covid-19 with new fractal-fractional
operators: can the lockdown save mankind before vaccination?
\newblock {\em Chaos, Solitons \& Fractals}, 136:109860, 2020.
	
\bibitem{atici2007transform}
F.~M. Atici and P.~W. Eloe.
\newblock A transform method in discrete fractional calculus.
\newblock {\em International Journal of Difference Equations}, 2(2), 2007.
	
\bibitem{bakare2014optimal}
E.~A. Bakare, A.~Nwagwo, and E.~Danso-Addo.
\newblock Optimal control analysis of an sir epidemic model with constant
recruitment.
\newblock {\em International Journal of Applied Mathematics Research},
3(3):273, 2014.
	
\bibitem{baleanu2010new}
D.~Baleanu, Z.~B. G{\"u}ven{\c{c}}, J.~T. Machado, et~al.
\newblock {\em New trends in nanotechnology and fractional calculus
applications}, volume~10.
\newblock Springer, 2010.
	
\bibitem{amarmat}
M.~Boukhobza, A.~Debbouche, L.~Shangerganesh, and J.~J. Nieto.
\newblock The stability of solutions of the variable-order fractional optimal
control model for the covid-19 epidemic in discrete time.
\newblock {\em Mathematics}, 12(8):1236, 2024.
	
\bibitem{bushnaq2023existence}
S.~Bushnaq, M.~Sarwar, H.~Alrabaiah, et~al.
\newblock Existence theory and numerical simulations of variable order model of
infectious disease.
\newblock {\em Results in Applied Mathematics}, 19:100395, 2023.
	
\bibitem{bushnaq2022computation}
S.~Bushnaq, K.~Shah, S.~Tahir, K.~J. Ansari, M.~Sarwar, and T.~Abdeljawad.
\newblock Computation of numerical solutions to variable order fractional
differential equations by using non-orthogonal basis.
\newblock {\em AIMS Mathematics}, 7(6):10917--10938, 2022.
	
\bibitem{debbouche1}
S.~A. David, C.~A. Valentim, and A.~Debbouche.
\newblock Fractional modeling applied to the dynamics of the action potential
in cardiac tissue.
\newblock {\em Fractal and Fractional}, 6(3):149, 2022.
	
\bibitem{gao2021stability}
F.~Gao, X.~Li, W.~Li, and X.~Zhou.
\newblock Stability analysis of a fractional-order novel hepatitis {B} virus
model with immune delay based on caputo-fabrizio derivative.
\newblock {\em Chaos, Solitons \& Fractals}, 142:110436, 2021.

\bibitem{ghaziani2016stability}
R.~K. Ghaziani, J.~Alidousti, and A.~B. Eshkaftaki.
\newblock Stability and dynamics of a fractional order leslie--gower
prey--predator model.
\newblock {\em Applied Mathematical Modelling}, 40(3):2075--2086, 2016.
	
\bibitem{huang2021discrete}
L.-L. Huang, G.-C. Wu, D.~Baleanu, and H.-Y. Wang.
\newblock Discrete fractional calculus for interval--valued systems.
\newblock {\em Fuzzy Sets and Systems}, 404:141--158, 2021.
	
\bibitem{kamyad2014mathematical}
A.~V. Kamyad, R.~Akbari, A.~A. Heydari, A.~Heydari, et~al.
\newblock Mathematical modeling of transmission dynamics and optimal control of
vaccination and treatment for hepatitis {B} virus.
\newblock {\em Computational and mathematical methods in medicine}, 2014, 2014.
	
\bibitem{khan2021stability}
A.~Khan, H.~M. Alshehri, T.~Abdeljawad, Q.~M. Al-Mdallal, and H.~Khan.
\newblock Stability analysis of fractional nabla difference covid-19 model.
\newblock {\em Results in Physics}, 22:103888, 2021.
	
\bibitem{khan2021transmission}
T.~Khan, Z.-S. Qian, R.~Ullah, B.~Al~Alwan, G.~Zaman, Q.~M. Al-Mdallal,
Y.~El~Khatib, and K.~Kheder.
\newblock The transmission dynamics of hepatitis {B} virus via the
fractional-order epidemiological model.
\newblock {\em Complexity}, 2021:1--18, 2021.
	
\bibitem{khan2016classification}
T.~Khan and G.~Zaman.
\newblock Classification of different hepatitis {B} infected individuals with
saturated incidence rate.
\newblock {\em SpringerPlus}, 5(1):1--16, 2016.
	
\bibitem{khan2017transmission}
T.~Khan, G.~Zaman, and M.~I. Chohan.
\newblock The transmission dynamic and optimal control of acute and chronic
hepatitis {B}.
\newblock {\em Journal of biological dynamics}, 11(1):172--189, 2017.
	
\bibitem{libbus2009public}
M.~K. Libbus and L.~M. Phillips.
\newblock Public health management of perinatal hepatitis {B} virus.
\newblock {\em Public health nursing}, 26(4):353--361, 2009.
	
\bibitem{lin2009stability}
R.~Lin, F.~Liu, V.~Anh, and I.~Turner.
\newblock Stability and convergence of a new explicit finite-difference
approximation for the variable-order nonlinear fractional diffusion equation.
\newblock {\em Applied Mathematics and computation}, 212(2):435--445, 2009.
	
\bibitem{debbouche}
J.~Manimaran, L.~Shangerganesh, A.~Debbouche, and J.~C. Cortés.
\newblock A time-fractional {HIV} infection model with nonlinear diffusion.
\newblock {\em Results in Physics}, 25:104293, 2021.
	
\bibitem{maynard1989global}
J.~E. Maynard, M.~A. Kane, and S.~C. Hadler.
\newblock Global control of hepatitis {B} through vaccination: role of
Hepatitis {B} vaccine in the expanded programme on immunization.
\newblock {\em Clinical Infectious Diseases}, 11(Supplement\_3):S574--S578,
1989.
	
\bibitem{medley2001hepatitis}
G.~F. Medley, N.~A. Lindop, W.~J. Edmunds, and D.~J. Nokes.
\newblock Hepatitis-{B} virus endemicity: heterogeneity, catastrophic dynamics
and control.
\newblock {\em Nature medicine}, 7(5):619--624, 2001.

\bibitem{nana2017hepatitis}
S.~Nana-Kyere, J.~Ackora-Prah, E.~Okyere, S.~Marmah, and T.~Afram.
\newblock Hepatitis {B} optimal control model with vertical transmission.
\newblock {\em Applied Mathematics}, 7(1):5--13, 2017.

\bibitem{onyango2014multiple}
N.~O. Onyango.
\newblock Multiple endemic solutions in an epidemic hepatitis {B} model without
vertical transmission.
\newblock {\em Applied Mathematics}, 5(16):2518--2529, 2014.
	
\bibitem{ortigueira2019variable}
M.~D. Ortigueira, D.~Val{\'e}rio, and J.~T. Machado.
\newblock Variable order fractional systems.
\newblock {\em Communications in Nonlinear Science and Numerical Simulation},
71:231--243, 2019.
	
\bibitem{petravs2011fractional}
I.~Petr{\'a}{\v{s}}.
\newblock {\em Fractional-order nonlinear systems: modeling, analysis and
simulation}.
\newblock Springer Science \& Business Media, 2011.
	
\bibitem{pol2005epidemiologie}
S.~Pol.
\newblock Epid{\'e}miologie et histoire naturelle de l’h{\'e}patite b.
\newblock {\em Rev Prat}, 55(6):599--606, 2005.
	
\bibitem{rachah2015mathematical}
A.~Rachah and D.~F.~M. Torres.
\newblock Mathematical modelling, simulation, and optimal control of the 2014
Ebola outbreak in west Africa.
\newblock Discrete Dyn. Nat. Soc. 2015 (2015), Art. ID 842792, 9~pp.
{\tt arXiv:1503.07396}
	
\bibitem{razminia2012solution}
A.~Razminia, A.~F. Dizaji, and V.~J. Majd.
\newblock Solution existence for non-autonomous variable-order fractional
differential equations.
\newblock {\em Mathematical and Computer Modelling}, 55(3-4):1106--1117, 2012.
	
\bibitem{rodrigues2014vaccination}
H.~S. Rodrigues, M.~T.~T. Monteiro, and D.~F.~M. Torres.
\newblock Vaccination models and optimal control strategies to dengue.
\newblock {\em Mathematical Biosciences}, 247:1--12, 2014.
{\tt arXiv:1310.4387}
	
\bibitem{rodrigues2014cost}
P.~Rodrigues, C.~J. Silva, and D.~F.~M. Torres.
\newblock Cost-effectiveness analysis of optimal control measures for
tuberculosis.
\newblock {\em Bulletin of mathematical biology}, 76:2627--2645, 2014.
{\tt arXiv:1409.3496}
	
\bibitem{safi2011effect}
M.~A. Safi and A.~B. Gumel.
\newblock The effect of incidence functions on the dynamics of a
quarantine/isolation model with time delay.
\newblock {\em Nonlinear Analysis: Real World Applications}, 12(1):215--235, 2011.
	
\bibitem{samko1993fractional}
S.~G. Samko.
\newblock Fractional integrals and derivatives.
\newblock {\em Theory and applications}, 1993.

\bibitem{doi:10.1080/10652469308819027}
S.~G. Samko and B.~Ross.
\newblock Integration and differentiation to a variable fractional order.
\newblock {\em Integral Transforms and Special Functions}, 1(4):277--300, 1993.
	
\bibitem{csener2018investigation}
A.~G. {\c{S}}ener, N.~Ayd{\i}n, C.~Ceylan, and S.~K{\i}rdar.
\newblock Investigation of antinuclear antibodies in chronic hepatitis {B}
patients.
\newblock {\em Mikrobiyoloji Bulteni}, 52(4):425--430, 2018.
	
\bibitem{shah2022spectral}
K.~Shah, H.~Naz, M.~Sarwar, and T.~Abdeljawad.
\newblock On spectral numerical method for variable-order partial differential
equations.
\newblock {\em AIMS Math}, 7(6):10422--10438, 2022.
	
\bibitem{shepard2006hepatitis}
C.~W. Shepard, E.~P. Simard, L.~Finelli, A.~E. Fiore, and B.~P. Bell.
\newblock Hepatitis {B} virus infection: epidemiology and vaccination.
\newblock {\em Epidemiologic reviews}, 28(1):112--125, 2006.
	
\bibitem{simelane2021fractional}
S.~Simelane and P.~Dlamini.
\newblock A fractional order differential equation model for hepatitis {B}
virus with saturated incidence.
\newblock {\em Results in Physics}, 24:104114, 2021.
	
\bibitem{soon2005variable}
C.~M. Soon, C.~F. Coimbra, and M.~H. Kobayashi.
\newblock The variable viscoelasticity oscillator.
\newblock {\em Annalen der Physik}, 517(6):378--389, 2005.
	
\bibitem{thornley2008hepatitis}
S.~Thornley, C.~Bullen, and M.~Roberts.
\newblock Hepatitis {B} in a high prevalence new zealand population: a
mathematical model applied to infection control policy.
\newblock {\em Journal of Theoretical Biology}, 254(3):599--603, 2008.
	
\bibitem{debbouche2}
S.~Tyagi, S.~C. Martha, S.~Abbas, and A.~Debbouche.
\newblock Mathematical modeling and analysis for controlling the spread of
infectious diseases.
\newblock {\em Chaos, Solitons \& Fractals}, 144:110707, 2021.
	
\bibitem{wang2010global}
K.~Wang, A.~Fan, and A.~Torres.
\newblock Global properties of an improved Hepatitis {B} virus model.
\newblock {\em Nonlinear Analysis: Real World Applications}, 11(4):3131--3138,
2010.
	
\bibitem{williams1996transmission}
J.~Williams, D.~Nokes, G.~Medley, and R.~Anderson.
\newblock The transmission dynamics of hepatitis {B} in the uk: a mathematical
model for evaluating costs and effectiveness of immunization programmes.
\newblock {\em Epidemiology \& Infection}, 116(1):71--89, 1996.
	
\bibitem{williams2006global}
R.~Williams.
\newblock Global challenges in liver disease.
\newblock {\em Hepatology}, 44(3):521--526, 2006.

\bibitem{xu2013existence}
Y.~Xu and Z.~He.
\newblock Existence and uniqueness results for cauchy problem of variable-order
fractional differential equations.
\newblock {\em Journal of Applied Mathematics and Computing}, 43:295--306,
2013.
	
\bibitem{yavuznew}
M.~Yavuz, F.~{\"O}zk{\"o}se, M.~Susam, and M.~Kalidass.
\newblock A new modeling of fractional-order and sensitivity analysis for
Hepatitis-{B} disease with real data, fractal fractional, 7 (2023), 165.
	
\bibitem{zhang2015modeling}
T.~Zhang, K.~Wang, and X.~Zhang.
\newblock Modeling and analyzing the transmission dynamics of hbv epidemic in
Xinjiang, China.
\newblock {\em PloS one}, 10(9):e0138765, 2015.
	
\bibitem{zhao2000mathematical}
S.~Zhao, Z.~Xu, and Y.~Lu.
\newblock A mathematical model of Hepatitis {B} virus transmission and its
application for vaccination strategy in China.
\newblock {\em International journal of epidemiology}, 29(4):744--752, 2000.
	
\bibitem{zou2010modeling}
L.~Zou, W.~Zhang, and S.~Ruan.
\newblock Modeling the transmission dynamics and control of Hepatitis {B} virus
in China.
\newblock {\em Journal of theoretical biology}, 262(2):330--338, 2010.
	
\end{thebibliography}
\end{document}